\documentclass{article}
\usepackage{ltexpprt}

\usepackage{dsfont} 	
\usepackage{amssymb} 	
\usepackage{bbm}		

\usepackage{mathtools}

\usepackage{algorithm}
\usepackage{algorithmic}

\usepackage{amsmath}

\usepackage{enumerate}

\usepackage{amsfonts}
\usepackage{amscd}

\usepackage{xcolor}

\usepackage{hyperref}

\newcommand{\EE}{\mathbb{E}}
\newcommand{\PP}{\mathbb{P}}

\usepackage[normalem]{ulem} 

\newtheorem{Corollary}[@theorem]{Corollary}
\newtheorem{Lem}[@theorem]{Lemma}

\newcommand{\qed}{\hfill \ensuremath{\Box}}

\begin{document}

\title{ 
\Large Prophet Secretary Through Blind Strategies 
\thanks{Partially supported by CONICYT-Chile through grant PII 20150140, By
ECOS-CONICYT thorough grant C15E03, and by a Google Research for Latin
America Award.
\newline
A preliminary version of this paper appeared in the Proceedings of the 30th ACM-SIAM Symposium on Discrete Algorithms (SODA 2019).} }
\author{ 
Jose Correa \thanks{Universidad de Chile, Santiago, Chile.}\\
\and
Raimundo Saona $^\dag$ \\ 
\and Bruno Ziliotto \thanks{CEREMADE, CNRS, Universit\'e Paris Dauphine, PSL University, Paris, France.}
}

\date{}

\maketitle




\begin{abstract} \small\baselineskip=9pt 
In the classic prophet inequality, a well-known problem in optimal stopping theory, samples from independent random variables (possibly differently distributed) arrive online. A gambler that knows the distributions, but cannot see the future, must decide at each point in time whether to stop and pick the current sample or to continue and lose that sample forever. The goal of the gambler is to maximize the expected value of what she picks and the performance measure is the worst case ratio between the expected value the gambler gets and what a prophet, that sees all the realizations in advance gets. In the late seventies, Krengel and Sucheston, and Garling \cite{KS77}, established that this worst case ratio is $1/2$.
A particularly interesting variant is the so-called prophet secretary problem, in which the only difference is that the samples arrive in a uniformly random order. For this variant several algorithms are known to achieve a constant of $1-1/e$ and very recently this barrier was slightly improved by  Azar et al. \cite{ACK18}. 

In this paper we introduce a new type of multi-threshold strategy, called \textit{blind strategy}. Such a strategy sets a nonincreasing sequence of thresholds that depends only on the distribution of the maximum of the random variables, and the gambler stops the first time a sample surpasses the threshold of the stage. 
Our main result shows that these strategies can achieve a constant of $0.669$ in the prophet secretary problem, improving upon the best known result of Azar et al. \cite{ACK18}, and even that of Beyhaghi et al. \cite{BGP18} that works in the case the gambler can select the order of the samples. The crux of the result is a very precise analysis of the underlying stopping time distribution for the gambler's strategy that is inspired by the theory of Schur convex functions. We further prove that our family of blind strategies cannot lead to a constant better than $0.675$. 

Finally we prove that no algorithm for the gambler can achieve a constant better than $\sqrt{3}-1$, which also improves upon a recent result of Azar et al. \cite{ACK18}. This implies that the upper bound on what the gambler can get in the prophet secretary problem is strictly lower than what she can get in the i.i.d. case. This constitutes the first separation between prophet secretary problem and the i.i.d. prophet inequality.

\end{abstract}

\section{Introduction}

{\bf Prophet-Inequalities.} For fixed $n>1$, let $V_1,\ldots, V_n$ be non-negative, independent random variables and $T_n$ the set of stopping times associated with the filtration generated by $V_1,\ldots, V_n$. A classic result of Krengel and Sucheston, and Garling \cite{KS77,KS78} asserts that 
$\EE(\max\{V_1,\ldots,V_n\})\le 2 \sup\{\EE (V_T):T\in T_n\}, $ and that 2 is the best possible bound. 
The interpretation of this result says that a {\em gambler}, who only knows the distribution of the random variables and that looks at them sequentially, can select a stopping rule that guarantees her half of the value that a {\em prophet}, who knows all the realizations could get.
The study of this type of inequalities, known as {\em prophet inequalities}, was initiated by Gilbert and Mosteller \cite{GM66} and attracted a lot of attention in the eighties \cite{HK82,K86,KW12,SC83}. In particular Samuel-Cahn \cite{SC83} noted that rather than looking at the set of all stopping rules one can (quite naturally) only look at threshold stopping rules in which the decision to stop depends on whether the value of the currently observed random variable is above a certain threshold (and possibly on the rest of the history). In the last decade the theory of prophet inequalities has resurged as an important problem due to its connections to \textit{posted price mechanisms} (PPMs) which are frequently used in online sales (see  \cite{CHM10} and \cite{EINAV}). The way these mechanisms work is as follows. Suppose a seller has an item to sell. Consumers arrive one at a time and the seller proposes to each consumer a take-it-or-leave-it offer. The first customer accepting the offer pays that price and takes the item. This is again a stopping problem, and we refer the reader to \cite{CHM10,CFPV18,HKS2007,M81} for the connection between this stopping problem and prophet inequalities.

Although the situation for the standard prophet inequality just described is well understood, there are variants of the problem, which are particularly relevant given the connection to PPMs, for which the situation is very different. In what follows we describe three important variants that are connected to each other and constitute the main focus of this paper.

\begin{itemize}

	\item {\em Order selection.} In this version the gambler is allowed to select the order in which she examines the random variables. For this version \cite{CHM10} improved the bound of $1/2$ (of the standard prophet inequality) to $1-1/e \approx 0.6321$. This bound remained the best known for quite some time until Azar et al. \cite{ACK18} improved it to $1-1/e+1/400\approx 0.6346$. Interestingly, the bound of Azar et al. actually applies to the random order case described below. Very recently Beyhaghi et al. \cite{BGP18}, use order selection to further improve the bound to $1-1/e+0.022\approx 0.6541$.
	
	\item {\em Prophet secretary (or random order).} In this version the random variables are shown to the gambler in random order, as in the secretary problem. This version was first studied by Esfandiari et al. \cite{EHL15} who found a bound of $1-1/e$. Their algorithm defines a nonincreasing sequence of $n$ thresholds $\tau_1,\ldots,\tau_n$ that only depend on the expectation of the maximum of the $V_i's$ and on $n$. The gambler at time-step $i$ stops if the value of $V_{\sigma_i}$ (the variable shown at step $i$) surpasses $\tau_i$. Later, Correa et al. \cite{CFH17} proved that the same factor of $1-1/e$ can be obtained with a personalized but nonadaptive sequence of thresholds, that is thresholds $\tau_1,\ldots,\tau_n$ such that whenever variable $V_i$ is shown the gambler stops if its value is above $\tau_i$. In recent work, Ehsani et al. \cite{EHK18} show that the bound of $1-1/e$ can be achieved using a single threshold (having to randomize to break ties in some situations). This result appears to be surprising since without the ability of breaking ties at random $1/2$ is the best possible constant and this insight turns out to be the starting point of our work. Shortly after the work of Ehsani et al., Azar et al. \cite{ACK18} improved it to $1-1/e+1/400\approx 0.6346$ through an algorithm that relies on some subtle case analysis. 

	\item {\em IID Prophet inequality.} Finally, we mention the case when the random variables are identically distributed. Here, the constant $1/2$ can also be improved. Hill and Kertz \cite{HK82} provided a family of ``bad'' instances from which  Kertz \cite{K86} proved the largest possible bound one could expect is $1/\beta\approx 0.7451$, where $\beta$ is the unique solution to $\int_0^1\frac{1}{y(1-\ln(y))+(\beta-1)}dy = 1$. Quite surprisingly, this upper bound is still the best known upper bound for the two variants above. Regarding algorithms Hill and Kertz also proved a bound of $1-1/e$ which was  improved by Abolhassan et al. \cite{AE17} to $0.7380$. Finally Correa et al. \cite{CFH17} proved that $1/\beta=0.7451$ is a tight value. To this end they exhibit a quantile strategy for the gambler in which some quantiles $q_1<\ldots < q_n$, that only depend on $n$ (and not on the distribution), are precomputed and then translated into thresholds so that if the gambler gets to step $i$, she will stop with probability $q_i$.
	
\end{itemize}


\subsection{Our contribution}

In this paper we study the prophet secretary problem and propose improved algorithms for it. In particular our work improves upon the recent work of Ehsani et al. \cite{EHK18}, Azar et al. \cite{ACK18}, and Beyhaghi et al. \cite{BGP18} by providing an algorithm that guarantees the gambler a fraction of $0.669 \approx 1-1/e+1/27$ in the prophet secretary setting. Our main contribution however is not the actual numerical improvement but rather the way in which this is obtained. In addition, we provide an example that shows that no algorithm can achieve a factor better than $\sqrt{3}-1 \approx 0.732$ for the prophet secretary setting.

From a conceptual viewpoint we introduce a class of algorithms which we call {\em blind} strategies, that are very robust in nature. This type of algorithm is a clever generalization of the single threshold algorithm of Ehsani et al. to a multi-threshold setting. In their algorithm Ehsani et al. first compute a threshold $\tau$ such that $\PP(\max\{V_1,\ldots,V_n\}\le \tau)=1/e$ and then use this $\tau$ as a single threshold strategy, so that the gambler stops the first time in which the observed value surpasses $\tau$. They observe that this strategy only works for random variables with continuous distributions, however they also note that by allowing randomization the strategy can be extended to general random variables. Rather than fixing a single probability of acceptance we fix a function $\alpha:[0,1]\to [0,1]$ which is used to define a sequence of thresholds in the following way. Given an instance with $n$ continuous distributions we draw uniformly and independently $n$ random values in $[0,1]$, and reorder them as $u_{[1]}<\cdots<u_{[n]}$. Then we set thresholds $\tau_1,\ldots,\tau_n$ such that $\PP(\max\{V_1,\ldots,V_n\}\le \tau_i)=\alpha(u_{[i]})$ and the gambler stops at time $i$ if $V_{\sigma_i}>\tau_i$. Note that if the function $\alpha$ is nonincreasing this will lead to a nonincreasing sequence of thresholds. 

The idea of blind strategies comes from the i.i.d. case mentioned above. In that setting the strategies are indeed best possible as shown by Correa et al. \cite{CFH17}. What makes blind strategies attractive is that although decisions are time dependent, this dependence lies completely in the choice of the function $\alpha$, which is independent of the instance. This independence significantly simplifies the analysis of multi-threshold strategies. Again, when facing discontinuous distributions we also require randomization for our results to hold. 

From a technical standpoint our analysis introduces the use of Schur convexity \cite{PPT92} in the prophet inequality setting. We start our analysis by revisiting the single threshold strategy of Ehsani et al., which corresponds to a constant blind strategy $\alpha(\cdot)=1/e$. We exhibit a new analysis of this strategy showing stochastic dominance type result. Indeed we prove that the probability that the gambler gets a value of more than $t$ is at least that of the maximum being more than $t$, rescaled by a factor $1-1/e$. This result uses Schur convexity to deduce that if there is a value above the threshold $\tau$, then it is chosen by the gambler with probability at least $1-1/e$. Then we extend this analysis to deal with more general functions $\alpha$ which require precise bounds on the distribution of the stopping time corresponding to a function $\alpha$. These bounds make use of results of Esfandiari et al. \cite{EHL15} and Azar et al. \cite{ACK18} and of newly derived bounds that follow from the core ideas in Schur convexity theory. 
	
Again in this more general setting we find an appropriate stochastic dominance type bound on the probability that the gambler obtains at least a certain amount with respect to the probability that the prophet obtains the same amount. Interestingly we manage to make the bound solely dependent on the blind strategy by basically controlling the implied stopping time distribution (patience of the gambler). Then optimizing over blind strategies leads to the improved bound of $0.669$. Through the paper we show two lower bounds on the performance of a blind strategy, the second more involved than the first one. In the first case, there is a natural way to optimize over the choice of $\alpha$ solving an ordinary differential equation, leading to a guarantee of $0.665$. In the second case, using a refined analysis, we derive the stated bound of $0.669$. Although it may seem that our general approach still leaves some room for improvement, we prove that  blind strategies cannot lead to a factor better than 0.675. This bound is obtained by taking two carefully chosen instances and proving that no blind strategy can perform well in both. 

Finally, we prove an upper bound on the performance of any algorithm: we construct an instance (which is not i.i.d.) in which no algorithm can perform better than $\sqrt{3}-1\approx 0.732$. This improves upon the best possible bound known of $0.745$ which corresponds to the i.i.d. case and was proved by Hill and Kertz \cite{HK82}. Furthermore it improves and generalizes a recent bound of  
$11/15 \approx 0.733$ of Azar et al. \cite{ACK18} for the more restricted class of {\em Deterministic distribution-insensitive algorithms}. Prior to our work, no separation between prophet secretary and i.i.d. prophet inequality was known. 

\subsection{Preliminaries and statement of results}
\label{prelim}
	Given nonnegative independent random variables $V_1,\ldots,V_n$ and a random permutation $\sigma:[n]\to[n]$\footnote{Here $[n]$ denotes the set $\{1,\ldots,n\}$}, in the prophet secretary problem a gambler is presented with the random variables in the order given by $\sigma$, i.e., at time $j$ she sees the realization of $V_{\sigma_j}$. The goal of the gambler is to find a stopping time $T$ such that $\EE(V_{\sigma_T})$ is as large as possible. In particular we want to find the largest possible constant $c$ such that 
	\[
	\sup\{\EE (V_{\sigma_T}):T\in \mathcal{T}_n\} \ge c\cdot \EE(\max\{V_1,\ldots,V_n\}) \,,
	\]
where $\mathcal{T}_n$ is the set of stopping times. 

Throughout this paper we denote by $F_1,\ldots,F_n$ the underlying distributions of $V_1,\ldots,V_n$, which we assume to be continuous. All our results apply unchanged to the case of general distributions by introducing random tie-breaking rules (this is made precise in Section \ref{discontinuous}). To see why random tie-breaking rules are actually needed, consider the single threshold strategy of Ehsani et al. \cite{EHK18}. Recall that they compute a threshold $\tau$ such that $\PP(\max\{V_1,\ldots,V_n\}\le \tau)=1/e$ and then use this $\tau$ as a single threshold strategy, which, by allowing random tie-breaking, leads to a  performance of $1 - 1/e$. However, if random tie-breaking is not allowed, a single threshold strategy cannot achieve a constant better than $1/2$. Indeed, consider the instance with $n-1$ deterministic random variables equal to $1$ and one random variable giving $n$ with probability $1/n$ and zero otherwise. Now, for a fixed threshold $\tau<1$ the gambler gets $n$ with probability $1/n^2$ and 1 otherwise so that she gets approximately $1$, whereas if $\tau\ge 1$ the gambler gets $n$ with probability $1/n$, leading to an expected value of $1$. Noting that the expectation of the maximum in this instance equals $2$, we conclude that fixed thresholds cannot achieve a guarantee better than $1/2$.  However, as Ehsani et al. note, if in this instance the gambler accepts a deterministic random variable with probability $1/n$, then her expected value approaches $2(1-1/e)$. In Section \ref{discontinuous} we extend this idea for the more general multi-threshold strategies.

The main type of stopping rules we deal with in this paper uses a nonincreasing threshold approach. This is a quite natural idea, since Esfandiari et al.\cite{EHL15} already use such an approach to derive a guarantee of $1-1/e$. Interestingly, the analysis of multi-threshold strategies becomes rather difficult when trying to go beyond this bound. This is evident from the fact that the more recent results take a different approach. In this paper we use a rather restrictive class of multi-threshold strategies that we call {\em blind} strategies. These are simply given by a nonincreasing function $\alpha: [0,1] \to [0,1]$ which is turned into an algorithm as follows: given an instance $F_1, ..., F_n$ of continuous distributions, we independently draw $u_1, ..., u_n$ from a uniform distribution on $[0, 1]$ and find thresholds $\tau_1, \ldots, \tau_n$ such that
	\[
	\mathbb{P}( \max\limits_{ i \in [n] } \{ V_i \} \leq \tau_i ) = \alpha( u_{[i]} ) \,,
	\]
where $u_{[i]}$ is the i-th order statistic of $u_1, \ldots, u_n$. Then the algorithm for the gambler stops at the first time in which a value exceeds the corresponding threshold, in other words the gambler applies the following.
\begin{algorithm}[H]
\caption{Time Threshold Algorithm ($TTA_{\tau_1, ..., \tau_n}$) }
\label{Alg: TTA}
\begin{algorithmic}[1]
	\FOR{ $i = 1, ..., n$ }
		\IF{ $V_{\sigma_i} > \tau_i$ }
			\STATE Take $V_{\sigma_i}$
		\ENDIF
	\ENDFOR
\end{algorithmic}
\end{algorithm}
Note that a blind strategy is uniquely determined by the choice of function $\alpha$, independent of the actual distributions or even size of the instance. This justifies that we may simply talk about strategy $\alpha$. Our goal is thus to find a {\em good} function $\alpha$ such that the previous algorithm performs well against any instance. 

For a blind strategy $\alpha$ and an instance $F_1,\ldots,F_n$, we will be interested in the underlying stopping time $T$ which is the random variable defined as $T := \inf \{ i \in [n] : V_{\sigma_i} > \tau_i \}$, where the $\tau_1,\ldots,\tau_n$ are the corresponding thresholds. In particular the reward of the gambler is $V_{\sigma_T} \mathds{1}_{T < \infty}$, which we simply denote by $V_{\sigma_T}$. 
\\
Our main result is the following:
\begin{theorem} \label{mainthm}
\label{Blind strategy factor} There exists a nonincreasing function $\alpha : [0, 1] \to [0, 1]$ such that
	\[
	\mathbb{E}( V_{\sigma_T}  ) 
		\geq 0.669 \;
			\mathbb{E}( \max\limits_{ i \in [n] } \{ V_i \} ) \,,
	\]
where $T$ is the stopping time of the blind strategy $\alpha$.
\end{theorem}
In addition, we prove the following upper bound on the performance of blind strategies:
\begin{theorem} \label{blindupper}
No blind strategies can guarantee a constant better than 0.675. 
\end{theorem}
Finally, we prove the following upper bound on the performance of general strategies:
\begin{theorem} \label{naupper}
No strategy can guarantee a constant better than $\sqrt{3}-1 \approx 0.732$. 
\end{theorem}
The rest of the paper is organized as follows. In Section \ref{sec:ST} we present an alternative simple proof that single threshold strategies guarantee a constant $1-1/e$, that will help the reader to understand the proof of 
Theorem \ref{mainthm}. Then, in Section \ref{beat}, we prove that blind strategies guarantee a constant 0.665. In Section \ref{blind detailed}, we sharpen the analysis of Section \ref{beat} to prove Theorem \ref{mainthm}, and also prove Theorem \ref{blindupper}. In Section
\ref{sec:naupper}, we prove the upper bound of Theorem  \ref{naupper}. Last, Section \ref{discontinuous} explains how to deal with discontinuous distributions. 
\section{Single Threshold}
\label{sec:ST}

As a warm-up exercise, we illustrate the main ideas in this paper by providing an alternative proof of a recent result by Eshani et al. \cite{EHK18}. Consider the blind strategy given by $\alpha\equiv p$, where $p\in [0,1]$ is a fixed number (taking $p=1/e$ gives exactly the single threshold algorithm of Eshani et al.). 

\begin{theorem}
\label{Fixed Threshold}
	 Let $t \geq 0$. Given $\alpha \equiv p$, 
	\[
	\mathbb{P}( V_{\sigma_T}  > t ) 
		\geq \min \left\{ 1 - p , \frac{1 - p}{-\ln p} \right\} 
			\mathbb{P}( \max\limits_{ i \in [n] } \{ V_i \}  > t ) \,.
	\]
\end{theorem}

\begin{proof} Recall that given an instance $V_1, \ldots, V_n$, the blind strategy $\alpha$ first computes $\tau$ such that
$\mathbb{P}( \max_{ i \in [n] } \{ V_i \}  \leq \tau ) = p$ and then uses $TTA_{\tau_1 = \tau, \ldots, \tau_n = \tau}$, which simply stops the first time a value above $\tau$ is observed.

Note that for $t \leq \tau$, we have that $\mathbb{P}( V_{\sigma_T} > t ) = \mathbb{P}( \max_{ i \in [n] } \{ V_i \} > \tau ) = 1 - p$. Now, for $t > \tau$, we have that
\begin{align*}
	\mathbb{P}( V_{\sigma_T} > t )
		&= \sum_{i \in [n]} \mathbb{P}( V_i > t | \sigma_T = i ) \mathbb{P}( \sigma_T = i ) \\
		&= \sum_{i \in [n]} \frac{ \mathbb{P}( V_i > t ) }{ \mathbb{P}( V_i > \tau ) }
			\mathbb{P}( \sigma_T = i ) \\
		&= \sum_{i \in [n]} \mathbb{P}( V_i > t ) 
			\mathbb{P}( \sigma_T = i |  V_i > \tau ) \\
		&\geq \left( \frac{ 1 - p }{ -\ln p } \right) \sum_{i \in [n]} \mathbb{P}( V_i > t ) \\
		&\geq \left( \frac{ 1 - p }{ -\ln p } \right) \mathbb{P}( \max_{i \in [n]} \{ V_i \} > t ).
\end{align*}
The second equality stems from the independence of the $V_i$, the first inequality is a consequence of Lemma \ref{Inequality: How likely is to pick you?} (which relies on Schur convexity), and the second inequality follows from the union bound.
\hfill \qed
\end{proof}

For a nonnegative random variable $V$ we have that $\EE(V)=\int_0^\infty \PP( V > t )dt$. Thus, an immediate consequence of Theorem \ref{Fixed Threshold} is a result of  Eshani et al. \cite{EHK18}.
\begin{Corollary}
[\cite{EHK18}] \label{onethreshold}
\label{Constant blind strategies factor} 
	Take $\alpha \equiv \frac{1}{e}$, then 
	$\mathbb{E}( V_{\sigma_T}  ) \geq \left( 1 - \frac{1}{e} \right)
			\mathbb{E}( \max_{ i \in [n] } \{ V_i \} )$.
\end{Corollary}

To complete the previous proof, we prove the following lemma.

\begin{Lem}
\label{Inequality: How likely is to pick you?}
Consider $V_1, \ldots, V_n$ independent random variables and $\sigma$ an independent random uniform permutation of $[n]$. Let $T$ be the stopping time of $TTA_{\tau_1 = \tau, \ldots, \tau_n = \tau}$ and $p = \mathbb{P}(\max_{i\in [n]} \{ V_i \} \leq \tau)$, then for all $i$ such that $\mathbb{P}( V_i > \tau ) > 0$ we have that
	\[
	\mathbb{P}( \sigma_T = i | V_i > \tau ) 
		\geq \frac{1 - p}{-\ln p} \,.
	\]
\end{Lem}

\begin{proof}
	Denoting the distribution of $V_j$ by $F_j$, using the fact that $\sigma$ is a uniform random order and the definition of $\tau$, we get the following.
\begin{align*}
	\mathbb{P}( \sigma_T = i | V_i > \tau )
		&= \sum\limits_{S \subseteq [n]\setminus\{i\}}
			\mathbb{P}( \sigma_T = i, S \cup \{ i \} = \{ j : V_j > \tau \} 
				| V_i > \tau ) \\
		&= \sum\limits_{S \subseteq [n]\setminus\{i\}}
			\frac{1}{|S| + 1}
			\prod\limits_{j \in S} 1 - F_j(\tau)
			\prod\limits_{j \in [n]\setminus(S \cup \{i\})} F_j(\tau) \\
		&= 	\prod\limits_{j \in [n]\setminus \{i\}} F_j(\tau)
			\sum\limits_{S \subseteq [n]\setminus\{i\}}
			\frac{1}{|S| + 1}
			\prod\limits_{j \in S} \frac{1 - F_j(\tau)}{F_j(\tau)}\\
		&= \frac{ p }{ F_i(\tau) }
			\sum\limits_{S \subseteq [n]\setminus\{i\}}
			\frac{1}{|S| + 1}
			\prod\limits_{j \in S} \frac{1 - F_j(\tau)}{F_j(\tau)} \,.
\end{align*}		
Define $\phi : \mathbb{R}_+^{n-1} \longrightarrow \mathbb{R}$ by
\begin{align*}
	\phi(y) :=  \sum\limits_{S \subseteq [n-1]}
				\frac{1}{|S| + 1}
				\prod\limits_{j \in S} \frac{1 - e^{-y_j}}{e^{-y_j}} 
			=  \sum\limits_{S \subseteq [n-1]}
				\frac{1}{|S| + 1}
				\prod\limits_{j \in S} e^{y_j} - 1
\end{align*}
and $\beta := -\ln p  + \ln F_i(\tau)$. Thus we have that
	\[
	\PP( \sigma_T = i | V_i > \tau ) \geq \frac{ p }{ F_i(\tau) } \min\left\{\phi(y) : \sum_{j \in [n-1]} y_j = \beta \right\} \,.
	\]
	
Clearly $\phi \in \mathcal{C}^{\infty}((0, 1)^{n-1} ; \mathbb{R})$ and is permutation symmetric. Therefore, to check that it is Schur convex we must simply confirm the following condition, known as the Schur-Ostrowski criterion \cite{PPT92},
	\[
	\forall y \in (0, 1)^{n-1} \quad 
		(y_1 - y_2)[\partial_{y_1}\phi(y) - \partial_{y_2}\phi(y)] \geq 0 \,.
	\]

Straightforward calculations yield 
\begin{align*}
	\partial_{y_1}\phi(y) 
		&= \sum\limits_{ \substack{ S \subseteq [n-1] \\ S \ni 1} }
			\frac{1}{|S| + 1} \quad
			e^{y_1} \prod\limits_{\substack { j \in S \\ j \not = 1} } e^{y_j} - 1 \\
		&= e^{y_1} (e^{y_1} - 1)^{-1} \left(
			\sum\limits_{ \substack{ S \subseteq [n-1] \\ S \ni 1, 2} }
			\frac{1}{|S| + 1} \quad
			 \prod\limits_{\substack { j \in S } } e^{y_j} - 1
			 \right) 
		+ \frac{e^{y_1}(e^{y_1} - 1)^{-1}}{(e^{y_2} - 1)}  \left( 
			\sum\limits_{ \substack{ S \subseteq [n-1] \\ S \ni 1, 2 } }
			\frac{1}{|S|} \quad
			\prod\limits_{\substack { j \in S } } e^{y_j} - 1 
			\right) \\
		&=: \frac{e^{y_1}}{e^{y_1} - 1 } a + \frac{e^{y_1}}{(e^{y_1} - 1)(e^{y_2} - 1)} b 
\end{align*}
and, by symmetry, $\partial_{y_2}\phi(y) = \frac{e^{y_2}}{e^{y_2} - 1 } a + \frac{e^{y_2}}{(e^{y_2} - 1)(e^{y_1} - 1)} b$. Then, 
\begin{align*}
	[\partial_{y_1}\phi(y) - \partial_{y_2}\phi(y)] 
		&= 	a 	\left[ 
				\frac{e^{y_1}}{e^{y_1} - 1 } 
				- \frac{e^{y_2}}{e^{y_2} - 1 } 
				\right] 
		+ b \left[ 
				\frac{e^{y_1}}{(e^{y_1} - 1)(e^{y_2} - 1)} 
				- \frac{e^{y_2}}{(e^{y_2} - 1)(e^{y_1} - 1)} 
				\right] \\
		&= (b - a) \left[ 
				\frac{e^{y_1} - e^{y_2}}{(e^{y_2} - 1)(e^{y_1} - 1)}  
				\right]\,.
\end{align*}
Finally, since $e^{y_j} - 1 > 0$ and $b > a$, we get that $(y_1 - y_2)[\partial_{y_1}\phi(y) - \partial_{y_2}\phi(y)] \geq 0$ if and only if $(y_1 - y_2)\left( e^{y_1} - e^{y_2} \right) \geq 0$, which holds by monotonicity of the exponential function. Therefore we have proven that $\phi$ is Schur-convex.

Schur convexity readily implies that the optimization problem $\min\left\{\phi(y);\, s.t. \sum_{j \in [n-1]} y_j = \beta \right\}$ is solved at $y^* = (\beta/(n-1), \ldots, \beta/(n-1))$. Consequently, for fixed $F_i(\tau)$, and under the constraint that $\prod_{j \in [n] \setminus \left\{i\right\}}F_j(\tau)=p/F_i(\tau)$, the quantity $\mathbb{P}( \sigma_T = i | V_i > \tau)$ is minimal when, for all $j \neq i$, $F_j(\tau)=(p/F_i(\tau))^{\frac{1}{n-1}}$.
\\
It follows that, since $\sigma$ and $V_i$ are independent,
\begin{align*}
	\mathbb{P}( \sigma_T = i | V_i > \tau )
		&= \frac{1}{n} \sum\limits_{j = 1}^{n} \mathbb{P}( \sigma_T = i | V_i > \tau, \sigma_j = i )
			\\
		&\geq \frac{1}{n} \sum\limits_{j = 0}^{n-1} \left( \frac{p}{F_i(\tau)} \right)^{ \frac{j}{n-1} } 
			\\
		&\geq \frac{1}{n} \sum\limits_{j = 0}^{n-1} p^{ \frac{j}{n-1} }
		= \frac{1}{n} \frac{ 1 - p^{ \frac{n}{n-1} } }{1 - p^{ \frac{1}{n-1} }}
			\,.
\end{align*}

	Now we note that the left hand side does not depend on $n$: we can add some dummy variables ($V_{n+1}, V_{n+2}, \ldots \equiv 0$) and the probability does not change. Therefore, taking limit on $n \longrightarrow \infty$ we get
	\[
	\mathbb{P}( \sigma_T = i | V_i > \tau )
		\geq \frac{1 - p}{-\ln p} \,.
	\]
\hfill \qed
\end{proof}

\section{Beating $1 - \frac{1}{e}$}
\label{beat}	
In the rest of the paper we assume for simplicity that $F_1,F_2,...,F_n$ are continuous (see Section \ref{discontinuous} for an explanation on how to extend the results to the discontinuous case). 
In this section, we prove the following proposition, that is a weaker version of Theorem \ref{mainthm} (the constant is 0.665 instead of 0.669):

\begin{proposition} \label{propsimple}
\label{Blind strategy factor} There exists a nonincreasing function $\alpha : [0, 1] \to [0, 1]$ such that
	\[
	\mathbb{E}( V_{\sigma_T}  ) 
		\geq 0.665 \;
			\mathbb{E}( \max\limits_{ i \in [n] } \{ V_i \} ) \,,
	\]
where $T$ is the stopping time of the blind strategy $\alpha$.
\end{proposition}
We first present this result because it already beats significantly the best known constant in literature, and it is simpler to prove than Theorem \ref{mainthm}. 
\\

	To do the analysis we first need to note that a blind strategy can be interpreted as the limit, as the size of the instance goes to infinity, of strategies that do not use randomization.
	 
\begin{Definition}
Consider a nonincreasing function $\alpha : [0, 1] \to [0, 1]$. The deterministic blind strategy given by $\alpha$ is the strategy that applies $TTA_{\tau_1, \ldots, \tau_n}$ to the sequence of thresholds $\tau_1, \ldots, \tau_n$ defined by the following conditions:
	\[
	\forall j \in [n], \quad \mathbb{P}( \max\limits_{ i \in [n] } \{ V_i \}  \leq \tau_j ) = \alpha\left( \frac{ j }{ n } \right) \,.
	\]
\end{Definition}	
	
To turn a deterministic blind strategy into a blind strategy, consider an instance $F_1, \ldots, F_n$ and add to this instance $m$ deterministic random variables equal to zero, that is, $F_{n+i}=\mathds{1}_{[0, \infty)}$ for $i=1,\ldots,m$, so that the new instance becomes $F_1, \ldots, F_n, F_{n+1}, \ldots, F_{n+m}$. Denoting by $T_m$ the stopping time given by the deterministic blind strategy applied to this instance, we have	that 
	\[
	\lim_{ m \to \infty } \mathbb{E}( V_{\sigma_{T_m} } ) 
		= \mathbb{E}( V_{\sigma_T} ) \,.
	\]
Indeed, recalling the definition of blind strategies in Section \ref{prelim}, it is easy to see that a deterministic blind strategy applied to instance $F_1, \ldots, F_{n+m}$ is approximately a blind strategy applied to instance $F_1, \ldots, F_n$, where the random variables $u_1, \ldots, u_n$ are drawn from $U( \{ \frac{1}{n+m}, \ldots, 1 \} )$, rather than from $U(0,1)$. Thus, by taking the limit as $m\to \infty$ the claim follows.
	
The conclusion of this remark is that in order to analyze the performance of blind strategies it is sufficient to study the performance of deterministic blind strategies and then take the limit as $n$ grows to infinity.

We are now ready to start analyzing deterministic blind strategies. 

\begin{Lem} 
\label{probab}
Given an instance $F_1,F_2,...,F_n$ and nonincreasing thresholds 
$\infty = \tau_0 \ge \tau_1 \ge \cdots \ge \tau_n \ge \tau_{n+1} = -\infty$, it holds that, for $j \in [n+1]$ and $t \in [\tau_j, \tau_{j-1})$,
\begin{align*}
	\mathbb{P}(V_{\sigma_T} &> t) 
		= \mathbb{P}(T \leq j-1 ) 
		+ \sum\limits_{i \in [n]} \mathbb{P}(V_i > t) \left(
			\sum\limits_{k > j-1 }^{n} \frac{ \mathbb{P}(T \geq k | \sigma_{k} = i) }{ n }
			\right)\,,
\end{align*}
where $T$ is the stopping time given by $TTA_{\tau_1, \ldots, \tau_n}$.
\end{Lem}
\begin{proof}
Notice that, since thresholds are nonincreasing, 
	\[
	\mathbb{P}(V_{\sigma_T} > t) 
		= \mathbb{P}(T \leq j-1 ) + \mathbb{P}(V_{\sigma_T} > t, T \geq j ) 
	\]
and therefore, we must simply analyze the second term.
\begin{align*}
	\mathbb{P}(V_{\sigma_T} > t, T \geq j ) 
		&= \sum\limits_{i \in [n]} 
			\mathbb{P}(V_i > t, \sigma_{T} = i, T \geq j ) \\
		&= \sum\limits_{i \in [n]} \sum\limits_{ k=j }^{ n } 
				\mathbb{P}(V_i > t, \sigma_{k} = i, T = k) \\
		&= \sum\limits_{i \in [n]} \sum\limits_{ k = j }^{n} 
				\mathbb{P}(V_i > t, \sigma_{k} = i, T \geq k) \\
		&= \sum\limits_{i \in [n]} \mathbb{P}(V_i > t) \left(
				\sum\limits_{k = j }^{n} \mathbb{P}(\sigma_{k} = i, T \geq k) 
				\right) \\
		&= \sum\limits_{i \in [n]} \mathbb{P}(V_i > t) \left(
				\frac{1}{n} \sum\limits_{ k = j }^{n} \mathbb{P}(T \geq k | \sigma_{k} = i) 
				\right) \,.
\end{align*}
\hfill \qed
\end{proof}
To give a lower bound on the right-hand side term, we use the following simple inequality: 
\begin{Lem} \label{simpleLem}
For all $i,k \in [n]$, 
\begin{equation} \label{simpleineq}
\mathbb{P}( T \geq k | \sigma_{k} = i) \geq \mathbb{P}( T > k ).
\end{equation}
\end{Lem}\begin{proof}
This is a particular case of Lemma \ref{Inequality: Conditional stopping time distribution}, that will be proved in the next section. 
\end{proof}
We can now minorize $\mathbb{P}(V_{\sigma_T} > t)$ by a quantity that depends only on the cumulative distribution of $T$ and on $\alpha$. 
\begin{proposition}
\label{blind}

	Let $\alpha : [0, 1] \to [0, 1]$ be nonincreasing, and let $T$ be the {\em deterministic blind strategy} stopping time. For every instance $F_1, \ldots, F_n$ and $t>0$,
	\[
	\mathbb{P}( V_{\sigma_T}  > t ) 
		\geq \min_{j \in [n+1]} \left\{ \frac{\mathbb{P}(T \leq j-1)}{1-\alpha(\frac{j}{n})}+\left( \frac{1}{n} \sum_{ k = j }^{n} \mathbb{P}(T > k ) \right)
			\right\} 
			\mathbb{P}( \max\limits_{ i \in [n] } \{ V_i \}  > t ) \,,
	\]
where $\alpha( \frac{n+1}{n} ) = 0$. 
	
\end{proposition}
\begin{proof}
Note that for all $j \in [n+1]$ and $t \in [\tau_j, \tau_{j-1})$, $1-\alpha(j/n) = \PP( \max_{i \in [n]} \{ V_i \} > \tau_j) \ge \PP( \max_{i \in [n]} \{ V_i \} > t)$.
Plugging this inequality and inequalities $\sum\limits_{i \in [n]} \mathbb{P}(V_i > t) \geq \mathbb{P}(\max\limits_{ i \in [n] } \{ V_i \}  > t )$ and (\ref{simpleineq}) into Lemma \ref{probab} yield the result. 
\end{proof}
We now bound the cumulative distribution of $T$ in function of the $\alpha(\frac{j}{n})$, $j \in [n+1]$. 
Both the lower and upper bounds are sharp in the sense that they are achieved by different instances: the lower bound corresponds to the case where there is only one non-zero variable and the upper bound corresponds to the case where all distributions are equal.

\begin{Lem}
\label{Inequality: Stopping time distribution}
	Fix $\alpha_1, \ldots, \alpha_n \in [0,1]$. For every instance $F_1, \ldots, F_n$ consider $\tau_1, \ldots, \tau_n$ the sequence of thresholds such that
	\[
	\mathbb{P}( \max_{ i \in [n] } \{ V_i \} \leq \tau_i ) = \alpha_i .
	\]
Denote by $T$ the stopping time of $TTA_{\tau_1, \ldots, \tau_n}$. For all $k \in [n]$, we have
	\[	 
		\frac{1}{n}\sum\limits_{j \in [k]} 1 - \alpha_j
		\leq \mathbb{P}( T \leq k )
		\leq 1 - \left( \prod_{j=1}^{k} \alpha_j \right)^{ \frac{1}{n} } \,.
	\]
\end{Lem}
\begin{proof}
	The proof consists in highlighting the role of $F_1$ and $F_2$ in $\mathbb{P}( T > k )$ and using the symmetry that the random order $\sigma$ induces. The key idea is to distinguish between the following cases: 
\begin{enumerate}
	\item $\sigma^{-1}(1) \leq k \, \veebar \, \sigma^{-1}(2) \leq k $, i.e. : only one of the variables $V_1$ and $V_2$ shows before time $k$.
	\item $\sigma^{-1}(1) \leq k \, \land \, \sigma^{-1}(2) \leq k $, i.e. : both $V_1$ and $V_2$ show before time $k$.
	\item $\sigma^{-1}(1) > k \, \land \, \sigma^{-1}(2) > k $, i.e. : neither $V_1$ nor $V_2$ shows before time $k$.
\end{enumerate}

	To express this formally, define
\begin{align*}
	\Sigma(k) &:= \{ \sigma, \text{ ordered subset of $[n]$ with size $k$} \} \\
	\Sigma_{-1, -2}(k) &:= \{ \sigma, \text{ ordered subset of $[n]\setminus\{1, 2\}$
		with size $k$} \} \,.
\end{align*} 
For $\sigma \in \Sigma(k)$, we have that either
\begin{enumerate}
	
	\item $\exists p \in [k], \exists i \in \{ 1, 2 \} \quad$ s.t. $\sigma_p = i$ and
		\[ 
		(\sigma_j)_{j \in [k]\setminus\{p\}} \in \Sigma_{-1, -2}(k-1) \,.
		\]
		
	\item $\exists p < q \in [k] \quad$ s.t. $\{ \sigma_p, \sigma_q \} = \{ 1, 2 \}$ and
		\[
		(\sigma_j)_{j \in [k]\setminus\{p, q\}} \in \Sigma_{-1, -2}(k-2) \,.
		\]

	\item $\sigma \in \Sigma_{-1, -2}(k)$.
\end{enumerate}
This is the key decomposition we use to show the inequality. In what follows, we write 
$\mathbb{P}( T > k ; F_1 , \ldots, F_n )$ for the probability that $T$ is strictly larger than $k$, given that the instance is $F_1 , \ldots, F_n$. 
We have
\begin{align*}
	\mathbb{P}( T > k ; F_1 , \ldots, F_n ) 
		&= \frac{1}{|\Sigma(k)|} \sum\limits_{\sigma \in \Sigma(k)} 
				\prod\limits_{i \in [k]} F_{\sigma_i}( \tau_i )\\
		&= \frac{(n-k)!}{n!} \left(
			\sum\limits_{\sigma \in \Sigma_{-1, -2}(k)} \right.
				\prod\limits_{i \in [k]} F_{\sigma_i}( \tau_i ) 
		+ \sum\limits_{ \substack{ \sigma \in \Sigma_{-1, -2}(k-1) \\ p \in [k] } } 
				\prod\limits_{i = 1}^{p-1} F_{\sigma_i}( \tau_i )
				\left[ F_{1}( \tau_p ) + F_{2}( \tau_p ) \right] 
		 \prod\limits_{i = p}^{k-1} F_{\sigma_i}( \tau_{i+1} ) \\
		&+ \sum\limits_{ \substack{ \sigma \in \Sigma_{-1, -2}(k-2) \\ p < q \in [k] } } 
				\prod\limits_{i = 1}^{p-1} F_{\sigma_i}( \tau_i ) 
			\left[ F_{1}( \tau_p )F_{2}( \tau_q ) + F_{2}( \tau_p )F_{1}( \tau_q ) \right]  
			\left. \prod\limits_{i = p}^{q-1} F_{\sigma_i}( \tau_{i+1} )
				\prod\limits_{i = q}^{k-2} F_{\sigma_i}( \tau_{i+2} ) \right) \,.
\end{align*}

	To simplify the notation, let us define
\begin{align*}
	&A(F_1, F_2):= 
		\sum_{ \substack{ \sigma \in \Sigma_{-1, -2}(k-1) \\ p \in [k] } }
				\left[ F_{1}( \tau_p ) + F_{2}( \tau_p ) \right]				
				\prod\limits_{i \in [k-1]} F_{\sigma_i}( \tau_{ i + \mathds{1}_{i \geq p} } )
				, \\ \\
	&B(F_1, F_2):= 
		\sum_{ \substack{ \sigma \in \Sigma_{-1, -2}(k-2) \\ p < q \in [k] } }
				\left[ F_{1}( \tau_p )F_{2}( \tau_q ) + F_{2}( \tau_p )F_{1}( \tau_q ) \right]
				\cdot \prod\limits_{i \in [k-2]} 
					F_{\sigma_i}( \tau_{ i + \mathds{1}_{i \geq p} + \mathds{1}_{i \geq q} } ), \\ \\
	&C:= \sum\limits_{\sigma \in \Sigma_{-1, -2}(k)} 
				\prod\limits_{i \in [k]} F_{\sigma_i}( \tau_i ) \,.
\end{align*}
Then, 
\begin{align*}
	\mathbb{P}&( T > k ; F_1 , \ldots, F_n ) 
		= \frac{(n-k)!}{n!} \left[ A( F_1, F_2 ) + B( F_1, F_2 ) + C \right] \,.
\end{align*}

	Let us show that both $A$ and $B$ change in the correct direction when we change $F_1$ and $F_2$, by $F_1 F_2$ and $\mathds{1}_{\mathbb{R}_+}$, or $\sqrt{F_1F_2}$ and $\sqrt{F_1F_2}$, respectively. For this, note that for all $p \in [k]$
\begin{align*}
	1 + F_1(\tau_p)F_2(\tau_p) 
		\geq F_1(\tau_p) + F_2(\tau_p) 
		&\geq 2 \sqrt{ F_1(\tau_p) F_2(\tau_p) } \,,
\end{align*}
and for all $p < q \in [k]$
\begin{align*}
	F_{1}( \tau_p )F_{2}( \tau_p ) + F_{2}( \tau_q )F_{1}( \tau_q ) 
		&\geq F_{1}( \tau_p )F_{2}( \tau_q ) + F_{2}( \tau_p )F_{1}( \tau_q ) \\
		&\geq 2 \sqrt{ F_{1}( \tau_p )F_{2}( \tau_p ) F_{1}( \tau_q )F_{2}( \tau_q ) } \,.
\end{align*}
Then, 
	\[
	A( F_1F_2, \mathds{1}_{\mathbb{R}_+} )
		\geq A( F_1, F_2 )
		\geq A( \sqrt{F_1 F_2}, \sqrt{F_1 F_2} ) \,,
	\]
and
	\[
	B( F_1F_2, \mathds{1}_{\mathbb{R}_+} )
		\geq B( F_1, F_2 ) 
		\geq B( \sqrt{F_1 F_2}, \sqrt{F_1 F_2} )\,.
	\]

	We can conclude on the lower bound by applying the inequality $n$ times and noticing that
	\[
	\mathbb{P} \left(
			T \leq k ; \prod_{ i \in [n] }F_i, \mathds{1}_{\mathbb{R}_+}, \ldots, \mathds{1}_{\mathbb{R}_+} 
			\right) 
		= \frac{1}{n}\sum\limits_{j \in [k]} 1 - \alpha_j \,.
	\]
	
	The upper bound follows from applying the inequality infinitely many times and noticing that
	\[
	\mathbb{P} \left(
			T \leq k ; \prod_{ i \in [n] }F_i^{ \frac{1}{n} }, \ldots, \prod_{ i \in [n] }F_i^{ \frac{1}{n} } 
			\right) 
		= 1 - \prod_{j \in [k]} \alpha_j^{ \frac{1}{n} } \,.
	\]
\hfill \qed
\end{proof}

Remember that in the proof of Lemma \ref{Inequality: How likely is to pick you?} we solved the following optimization problem: $\min\left\{\phi(y);\, s.t. \sum_{j \in [n-1]} y_j = \beta \right\}$. The value of this problem was obtained by noticing that $\phi$ is Schur convex. This time we considered the problem 
	\[
	\left\{ \begin{array}{l l}
		\text{opt } \mathbb{P}( T > k ; F_1, \ldots, F_n ) \\
		s.t. \prod_{ i \in [n] } F_i = F \text{ and $F_i$ is a distribution}.
	\end{array}
	\right.
	\]
where ``opt'' is a symbol in $\{ \min, \max \}$. This problem is harder since it involves optimizing over functions rather than real numbers. Trying to apply Schur convexity theory again, one could see $\mathbb{P}( T > k ; F_1, \ldots, F_n )$ as a function of the distributions evaluated at each threshold, that is as a function of the vector $(F_1(\tau_1), \ldots, F_1(\tau_n), \ldots, F_n(\tau_1), \ldots, F_n(\tau_n))$. Unfortunately this domain is not symmetric and moreover the constraint of the product being constant results in $n$ different constraints. 

However, note that the previous lemma shows that $\mathbb{P}( T > k ; F_1, \ldots, F_n )$ is nearly log-Schur-convex: it increases when the components of the argument get more concentrated in some coordinate. Nevertheless, the behavior of $\mathbb{P}( T > k )$ is not always monotone along the curve $\lambda \in [0, 1] \mapsto (F_1 F_2^{\lambda}, F_2^{1-\lambda} ,\ldots,F_n) $, a property that would be satisfied by a log-Schur-convex function if $F_1, \ldots, F_n$ were numbers. In spite of the latter, there is a step by step way to go from $(F_1, F_2,\ldots,F_n)$ to $(F_1 F_2 \ldots F_n, \mathds{1}_{\mathbb{R}_+}, \ldots, \mathds{1}_{\mathbb{R}_+})$ that exhibits a monotonic behavior, while maintaining the product. This property could be called {\em weak} log-Schur-convexity and is enough to solve the optimization problem. The same can be said about the points $(F_1, F_2,\ldots,F_n)$ and $( \sqrt[n]{ F_1 F_2 \ldots F_n}, \ldots, \sqrt[n]{ F_1 F_2 \ldots F_n} )$. 
\\

We can now give a lower bound on $\mathbb{P}(V_{\sigma_T} > t)$ that only depends on $\alpha$.  

\begin{theorem}
\label{blind}

	Let $\alpha : [0, 1] \to [0, 1]$ be nonincreasing, and let $T$ be the {\em deterministic blind strategy} stopping time. For every instance $F_1, \ldots, F_n$ and $t>0$,
	\[
	\mathbb{P}( V_{\sigma_T}  > t ) 
		\geq \min_{j \in [n+1]} \left\{ f_j( \alpha ) \right\} 
			\mathbb{P}( \max\limits_{ i \in [n] } \{ V_i \}  > t ) \,,
	\]
where, for all $j \in [n+1]$, taking $\alpha( \frac{n+1}{n} ) = 0$,
	\[
	f_j( \alpha ) =   \sum\limits_{k = 1}^{j-1} 
						\frac{ 1 - \alpha \left( \frac{ k }{ n } \right) }
							{ n \left( 1 - \alpha \left( \frac{ j }{ n } \right) \right) }
		+ \frac{1}{n} \sum\limits_{ k = j }^{ n } 
			\left( \prod\limits_{l=1}^{k} \alpha \left( \frac{ l }{ n } \right) \right) ^ { \frac{1}{n} } \,.
	\]
\end{theorem}

\begin{proof}
This is a direct consequence of Proposition \ref{blind} and Lemma \ref{Inequality: Stopping time distribution}. 
%
\hfill \qed
\end{proof}

Thus, for every $n$, we get a lower bound on the performance of a deterministic blind strategy $\alpha$, that only depends on $\alpha( \frac{1}{n} ), \ldots, \alpha( \frac{n}{n} )$. As we explained before, we only care about the performance of this strategy when $n$ tends to $+\infty$. Assume that $\alpha$ is continuous. A standard Riemann sum analysis shows that
\begin{align}
\label{inf} \lim_{n \to \infty} &\min_{j \in [n+1]} \{ f_j( \alpha ) \} 
		= \min \left\{ \int_{ 0 }^{ 1 } 1 - \alpha( y ) dy  \, , \inf_{ x \in [0, 1]} 
			\int_{ 0 }^{ x } \frac{ 1 - \alpha( y )}{ 1 - \alpha( x ) } dy
			+ \int_{ x }^{ 1 } e ^ { \int_{ 0 }^{ y } \ln \alpha(w) dw } dy \right\} \,.
\end{align}

Thus, in order to prove Theorem \ref{Blind strategy factor}, we would like to find a blind strategy $\alpha$ maximizing the latter expression. As this is a nontrivial optimal control problem we aim at finding a function $\alpha$ such that the above expression is larger than $0.665$. 

\textbf{Remark.} Consider $\alpha$ being constant equal to $1/e$. Then the above quantity is equal to $1-1/e$. Thus, we recover the one-threshold result in Corollary \ref{onethreshold}. Furthermore, if for instance we take $\alpha( x ) = 0.53-0.38x$ the guarantee of the strategy (given by expression \eqref{inf}) is greater than $0.657$. This gives an explicit $\alpha$ that beats significantly $1-1/e$.

To maximize over expression \eqref{inf}, we resort to a numerical approximation. Note that if $\alpha$ is such that $\alpha(1)=0$ and $x \mapsto \int\limits_{ 0 }^{ x } \frac{ 1 - \alpha( y )}{ 1 - \alpha( x ) } dy + \int\limits_{ x }^{ 1 } e ^ { \int\limits_{ 0 }^{ y } \ln \alpha(w) dw } dy$ is a constant, then this constant is a lower bound for the infimum in \eqref{inf}. 

Consequently, we solve the following integro-differential equation:
	\[
	\left\{ \begin{array}{l r}
	\frac{d}{dx} \left( 
		\int\limits_{ 0 }^{ x } \frac{ 1 - \alpha( y )}{ 1 - \alpha( x ) } dy
			+ \int\limits_{ x }^{ 1 } e ^ { \int\limits_{ 0 }^{ y } \ln \alpha(w) dw } dy
		\right) 
			= 0
			& ; x \in (0, 1)\\ 
	\alpha(1) = 0.
	\end{array} \right. 
	\]	
To this end we consider a change of variables leading to the following second order ODE:
	\[
	\left\{ \begin{array}{l r}
	(u'(x))^2 K(x, u)
		- u''(x) u(x) = 0
			& ; x \in (0, 1)\\ 
	u'(1) = 1 \\
	u(0) = 0,
	\end{array} \right. 
	\]
where
\begin{align*}
	u(x) &:= \int\limits_{0}^{ x } 1 - \alpha(x) dx \\
	K(x, u) &:= 1 - \exp \left( \int\limits_{0}^{x} \ln( 1 - u'(t) ) dt \right) \,.
\end{align*}
We approximately solved this equation by taking an initial guess $u_0$ and defining $u_{n+1}$ as the solution to $(u'(x))^2 K(x, u_n) - u''(x) u(x) = 0$. To be more precise, the initial guess $u_0$ was the result of maximizing over $\alpha$ $\min_{j \in [n+1]} \{ f_j(\alpha) \}$, given in Theorem \ref{blind}, for $n=23$. Then, we iterated the process eleven times and obtained an $\alpha$ with $\alpha(1) = 0$ and such that the function $x \mapsto \int\limits_{ 0 }^{ x } \frac{ 1 - \alpha( y )}{ 1 - \alpha( x ) } dy + \int\limits_{ x }^{ 1 } \exp\left({ \int\limits_{ 0 }^{ y } \ln \alpha(w) dw}\right) dy$ varies between $0.6653$ and $0.6720$. Even if we did not find an exact solution for the ODE, its performance is given by computing \eqref{inf}. This gives the claimed factor of $0.665$.

\section{Improved analysis and proof of Theorem \ref{mainthm}}
\label{0669}
\label{blind detailed}
In this section, we present how to get the factor 0.669 in Theorem \ref{mainthm}. The general method is the same, and the key ingredient is to use a sharper inequality than (\ref{simpleineq}), that is written below. 
\begin{Lem}
\label{Inequality: Conditional stopping time distribution}
	Given $V_1, \ldots, V_n$ independent random variables and $\tau_1 \geq \ldots \geq \tau_n$ a sequence of nonincreasing thresholds, we denote by $T$ the stopping time of $TTA_{\tau_1, \ldots, \tau_n}$. Then, $\forall i, k \in [n]$
	\[
	\mathbb{P}( T \geq k | \sigma_{k} = i)
		\geq \frac{ \mathbb{P}( T > k ) }{
			1 - \frac{k}{n} + 
			\frac{1}{n} \sum\limits_{l \in [k]}
				\mathbb{P}( V_{i} \leq \tau_l ) 
			} \,.
	\]
\end{Lem}

\begin{proof}
	Inspired by the proof given by Esfandiari et al. \cite{EHL15}, fix $i, k \in [n]$ and define
\begin{align*}
	\Sigma_{-i}(k) &:= \{ \sigma, \text{ ordered subset of } [n]\setminus \{ i \} \text{  with size $k$} \} \,.
\end{align*}
Then, conditioning on the value of $\sigma^{-1}(i)$,
\begin{equation}
\label{distribution}
	\mathbb{P}( T > k )
			= \frac{1}{n} \sum\limits_{l \in [k]} \mathbb{P}( T > k | \sigma_l = i ) 
	+ \frac{1}{n} \sum\limits_{l = k+1}^{n} \mathbb{P}( T > k| \sigma_l = i ) \,.
\end{equation}
Note that for all $l = k+1, \ldots, n$, 
\begin{align*}
	\mathbb{P}( T > k| \sigma_l = i ) 
		= \mathbb{P}( T > k| \sigma_{k+1} = i ) 
		&\leq \mathbb{P}( T \geq k| \sigma_k = i )
\end{align*}
and given $l \in [k]$, since the thresholds are nonincreasing,
\begin{align*}
	\mathbb{P}( T > k | \sigma_l = i ) 
		&= \frac{ \mathbb{P}( V_i \leq \tau_l) }{|\Sigma_{-i}(k-1)|} \sum\limits_{ \sigma \in \Sigma_{-i}(k-1) }
				\prod_{ j \in [k-1] } F_{\sigma_j}( \tau_{j + \mathds{1}_{ j \geq l}} ) \\
		&\leq  \frac{ \mathbb{P}( V_i \leq \tau_l) }{|\Sigma_{-i}(k-1)|} \sum\limits_{ \sigma \in \Sigma_{-i}(k-1) }
				\prod_{ j \in [k-1] } F_{\sigma_j}( \tau_{j} ) \\
		&= \mathbb{P}( V_i \leq \tau_l)
				\mathbb{P}( T \geq k | \sigma_{k} = i ) \,.		
\end{align*}

Plugging both inequalities back into equation \eqref{distribution} we get the result.

\hfill \qed
\end{proof}
Recall that in the proof of Proposition \ref{propsimple}, we used the inequality $\sum_{i \in [n]} \mathbb{P}(V_i>t) \geq \mathbb{P}( \max\limits_{ i \in [n] } \{ V_i \} > t) $. We will use now the following more sophisticated inequality: 
\begin{Lem}
\label{Inequality: Concave or convex?}
	Let $V_1, ..., V_n$ independent random variables and $\tau_1 \geq \max \{ \tau_2, \tau_3, ..., \tau_n \}$ a sequence of thresholds. Then, for any $ t < \tau_1$ and $k \leq \frac{n}{2}$,
	\[
	\sum\limits_{ i \in [n] } 
		\frac{ \mathbb{P}( V_i > t ) }{ 1 - \frac{1}{n} \sum\limits_{ l \in [k]} \mathbb{P}( V_i > \tau_l) }
			\geq \frac{ \mathbb{P}( \max\limits_{ i \in [n] } \{ V_i \} > t ) }
					{ 1 - \frac{k}{n} \mathbb{P}( \max\limits_{ i \in [n] } \{ V_i \} > \tau_1) } \,.
	\]
\end{Lem}

\begin{proof}
	Define $\lambda := \frac{k}{n} \in [0, \frac{1}{2}]$ and notice that
\begin{align*}
	\sum\limits_{ i \in [n] } 
		\frac{ \mathbb{P}( V_i > t ) }{ 1 - \frac{1}{n} \sum\limits_{ l \in [k]} \mathbb{P}( V_i > \tau_l) } 
		&= \sum\limits_{ i \in [n] } 
			\frac{ \mathbb{P}( V_i > t ) }
				{ \frac{n - k}{n} + \frac{1}{n} \sum\limits_{ l \in [k]} \mathbb{P}( V_i \leq \tau_l) } \\
		&\geq \sum\limits_{ i \in [n] } 
			\frac{ \mathbb{P}( V_i > t ) }
				{ 1 - \lambda + \lambda \mathbb{P}( V_i \leq \tau_1) } \\
		&= \sum\limits_{ i \in [n] } 
			\frac{ 1 - F_i( t ) }
				{ 1 - \lambda + \lambda F_i( \tau_1) } \\
		&=: C( t ; \lambda, F_1, \ldots, F_n ) \,.		
\end{align*}

	Therefore, it is sufficient to prove that
\begin{equation}
\label{F inequality}
	\frac{ 1 - F_1( t ) }{ 1 - \lambda + \lambda F_1( \tau_1) } 
	+ \frac{ 1 - F_2( t ) }{ 1 - \lambda + \lambda F_2( \tau_1) } 
	\geq \frac{ 1 - F_1( t ) F_2( t ) }{ 1 - \lambda + \lambda F_1( \tau_1) F_2( \tau_1) } \,,
\end{equation}
since we can iterate this argument $n$ times to deduce the result. To prove this inequality, define the following variables
\begin{align*}
	&\beta	:= F_1( \tau_1 ) F_2( \tau_1 ) , 
	&&\gamma 	:= F_1( t ) F_2( t ) , \\
	&x		:= F_1( \tau_1 ) , 
	&&y 		:= F_1( t ) \,.
\end{align*}
We can now simply solve the following optimization problem:
	\[
	(P) \left\{ \begin{array}{l c}
	\min_{x, y} & f_{ \lambda }( x, y ) := \frac{ 1 - y }{ 1 - \lambda + \lambda x } 
							+ \frac{ 1 - \frac{ \gamma}{y} }{ 1 - \lambda + \lambda \frac{ \beta }{x} } \\
	s.t. 		& \beta	\leq x \leq 1\\
				& \beta \leq \frac{ \beta }{ x } \leq 1 \\
				& \gamma \leq y \leq x \\
				& \gamma \leq \frac{ \gamma }{ y } \leq \frac{ \beta }{ x } \,.
	\end{array} 
	\right.
	\]
		
	Notice that the function $f_{\lambda}$ defined by
\begin{align*}
	f_{\lambda}(x, y) = \frac{ 1 }{ 1 - \lambda + \lambda x } 
			- \left( \frac{ 1 }{ 1 - \lambda + \lambda x } \right) y 
			+ \frac{ 1 }{ 1 - \lambda + \lambda \frac{ \beta }{x} }
			- \left( \frac{ \gamma }{ 1 - \lambda + \lambda \frac{ \beta }{x} } \right) \frac{1}{y} \,,
\end{align*}
is concave in $y$. Rearranging the inequalities, the problem reduces to
	\[
	(P) \left\{ \begin{array}{l c}
	\min_{x} & f_{ \lambda }( x ) := \frac{ 1 - x }{ 1 - \lambda + \lambda x } 
							+ \frac{ 1 - \frac{ \gamma}{x} }{ 1 - \lambda + \lambda \frac{ \beta }{x} } \\
	s.t. 		& \beta	\leq x \leq 1
	\end{array} \right. \,
	\]

	Then, all we have to show is that, if $\lambda \in [0, \frac{1}{2}]$, $x = 1$ is the minimum, since this would imply inequality \eqref{F inequality}. The following are three sufficient conditions for $x=1$ to be the minimum:
\begin{enumerate}
	\item $f_{\lambda}( \beta ) \geq f_{\lambda}( 1 )$.
	\item $x = 1$ is local minimum. 
	\item There exists at most one critical point in the interval $[\beta, 1]$.
\end{enumerate}

	All these conditions are true for $\lambda \in [0, 1/2]$, so we can prove inequality \eqref{F inequality}. Therefore, the lemma is proved iterating this inequality $n$ times.
\hfill \qed
\end{proof}

To prove Theorem \ref{Blind strategy factor} (i.e., to improve upon the bound proved in the last section) we turn to a particular type of blind strategies where $\alpha = \alpha_{\alpha_1, \ldots, \alpha_m}$ is given by
	\[
	\alpha_{\alpha_1, \ldots, \alpha_m}( x )
		= \sum\limits_{ j \in [m] } \alpha_j \mathds{1}_{ \left[ \frac{j-1}{m}, \frac{j}{m} \right) }( x ) \,,
	\]
in other words, piece-wise constant functions.
	
The idea now is to use the property that, for any instance of size larger than $m$, this blind strategy uses only $m$ thresholds. 
	
	Define for $m \geq 1$
	\[
	g_{ m , p }( k ) = \left\{ \begin{array}{l l}
		\frac{ 1 }{ 1 - \frac{k}{m} ( 1 - p ) }
			&; k \leq \frac{m}{2} \\
		\frac{ 2 }{ 1 + p }
			&; k > \frac{m}{2} \,.
		\end{array}	 
		\right.
	\]
This function is used in the next lemma. Notice that it is a nondecreasing function in $k$ that is always greater than $1$. We can prove the following statement, which has the same flavor as Theorem \ref{blind}.

\begin{Lem}
\label{Performance guarantee for piece-wise constant blind strategies}
	Let $ \alpha = \alpha_{\alpha_1, \ldots, \alpha_m}$ be a nonincreasing function where $\alpha_m > 0$, and let $T$ be the blind strategy stopping time. Then, 
	\[
	\frac{ \mathbb{E}( V_{\sigma_T} ) }{ \mathbb{E}( \max\limits_{ i \in [n] } \{ V_i \} ) }
		\geq \min_{ j \in [m+1] } \{ f_j( \alpha_1, \ldots, \alpha_m ) \} \,,
	\]
where $ $ is given by
	\[
	f_j( \alpha_1, \ldots, \alpha_m ) = \left\{ \begin{array}{l l}
		\sum\limits_{ k = 1 }^{ m } 
			\left( \prod\limits_{ l \in [k-1] } \alpha_l \right)^{ \frac{1}{m} } 
			\left( \frac{ 1 - \alpha_k^{ \frac{1}{m} } }{ -\ln \alpha_k } \right)
			&; j = 1 \\
		\sum\limits_{ k \in [j-1] } \frac{ 1 - \alpha_k }{ m( 1 - \alpha_j ) } 
		+ \sum\limits_{ k = j }^{ m } 
		\left( \prod\limits_{ l \in [k-1] } \alpha_l \right)^{ \frac{1}{m} } 
		g_{ m, \alpha_1 }( k - 1 ) 
		\left( \frac{ 1 - \alpha_k^{ \frac{1}{m} } }{ -\ln \alpha_k } \right)
			&; j \in \{ 2, \ldots, m \} \\
		\frac{1}{m} \sum\limits_{ k \in [m] } 1 - \alpha_k
			&; j = m+1 \,,
		\end{array}
		\right.
	\]

\end{Lem}

\begin{proof}
	As done in section \ref{beat}, we analyze the performance of the corresponding deterministic blind strategy with an instance of size $n$ and we only care about the performance guarantee of $\alpha$ as $n$ grows to $\infty$. For $N \geq 1$, consider an instance $F_1, \ldots, F_{Nm}$, and take $j \in [m+1]$ and $t \in [ \tau_{j}, \tau_{j-1} )$, where $\tau_0 = \infty$, $\tau_{m+1} = 0$ and $\alpha_{m+1} = 0$. In the same spirit as in Lemma \ref{probab}, we have that
\begin{align*}
	\mathbb{P}( V_{\sigma_T} > t ) 
		&= \mathbb{P}( T \leq ( j - 1 )N ) 
			+ \mathbb{P}( V_{\sigma_T} > t, T > ( j - 1 )N ) \\
		&= \mathbb{P}( T \leq ( j - 1 )N ) 
			+  \sum\limits_{ i \in [n] } \mathbb{P}( V_i > t )
				\left( \sum\limits_{ k = ( j - 1 )N + 1 }^{ Nm }
				\frac{ \mathbb{P}( T \geq k | \sigma_k = i ) }{ Nm }
				\right) \,.
\end{align*}

One key point is that since the same thresholds are used $N$ times, we can deduce better bounds. For the first term, as before, we have that

\begin{align*}
	\mathbb{P}( T \leq ( j - 1 )N ) 
		\geq \frac{ \mathbb{P}( T \leq ( j - 1 )N ) }{ 1 - \alpha_j }
				\mathbb{P}( \max\limits_{ i \in [n]} \{ V_i \} > t ) \,.
\end{align*}

Then, for $j = m+1$ (i.e. : $t \in [0, \tau_m)$), using Lemma \ref{Inequality: Stopping time distribution}, we have that
\begin{align*}
	\mathbb{P}( V_{\sigma_T} > t )
		\geq \frac{ 1 }{ m } \sum\limits_{ k \in [m] } 1 - \alpha_k \,,
\end{align*}
which concludes the case $j = m+1$, since $\alpha_{m+1} = 0$.
	
Now, for $j \in \{ 2, \ldots, m \}$, we must show the lower bound for
\begin{align*}
	&\mathbb{P}( T \leq ( j - 1 )N ) 
		+  \sum\limits_{ i \in [n] } \mathbb{P}( V_i > t )
			\sum\limits_{ k = ( j - 1 )N + 1 }^{ Nm }
			\frac{ \mathbb{P}( T \geq k | \sigma_k = i ) }{ Nm } \,.
\end{align*}
For the first term we use again Lemma \ref{Inequality: Stopping time distribution} and deduce that 
\begin{align*}
	\mathbb{P}( T \leq ( j - 1 )N ) 
		&\geq \sum_{ k \in [( j - 1 )N] } \frac{ 1 - \alpha( k/Nm ) }{ Nm } \\
		&= \sum_{ k \in [ j-1 ] } \frac{ 1 - \alpha_k }{ m } \\
		&\geq \sum_{ k \in [ j-1 ] } \frac{ 1 - \alpha_k }{ m ( 1 - \alpha_j ) } 
			\mathbb{P}( \max_{ i \in [n] } \{ V_i \} > t) \,.
\end{align*}
Noticing that $\alpha$ is nonincreasing, the corresponding thresholds are nonincreasing and we can use both Lemma \ref{Inequality: Conditional stopping time distribution} and Lemma \ref{Inequality: Concave or convex?} in the following way. First, for every $i, k \in [Nm]$,
	\[
	\mathbb{P}( T \geq k | \sigma_k = i ) 
		\geq \frac{ \mathbb{P}( T > k) }{ 1 - \frac{1}{Nm} \sum_{ l \in [k] } \mathbb{P}( V_i > \tau_l ) } \,.
	\]
Then, interchanging the order of the sums,
\begin{align*}
	\sum\limits_{ i \in [n] } \mathbb{P}( V_i > t )
		\sum\limits_{ k = ( j - 1 )N + 1 }^{ Nm }
		\frac{ \mathbb{P}( T \geq k | \sigma_k = i ) }{ Nm } 
			\geq \sum\limits_{ k = ( j - 1 )N + 1 }^{ Nm } \mathbb{P}( T > k )
				\sum\limits_{ i \in [n] } 
					\frac{ \mathbb{P}( V_i > t )  }{ Nm - \sum_{ l \in [k] } \mathbb{P}( V_i > \tau_l ) } \,.
\end{align*}
Now, by Lemma \ref{Inequality: Concave or convex?}, for $ k \leq Nm/2$,
\begin{align*}
	\sum\limits_{ i \in [n] } 
		\frac{ \mathbb{P}( V_i > t )  }{ Nm - \sum_{ l \in [k] } \mathbb{P}( V_i > \tau_l ) } 
			&\geq \frac{ \mathbb{P}( \max_{ i \in [n] } \{ V_i \} > t )  }
				{ Nm - k \mathbb{P}( \max_{ i \in [n] } \{ V_i \} > \tau_1 ) } \\
			&= \frac{ g_{ Nm, \alpha_1 }( k ) }{Nm} \mathbb{P}( \max_{ i \in [n] } \{ V_i \} > t ) \,,
\end{align*}
and for $k > Nm/2$, 
\begin{align*}
	\sum\limits_{ i \in [n] } 
		\frac{ \mathbb{P}( V_i > t )  }{ Nm - \sum_{ l \in [k] } \mathbb{P}( V_i > \tau_l ) } 
			&\geq \frac{ \mathbb{P}( \max_{ i \in [n] } \{ V_i \} > t )  }
				{ Nm - Nm/2 \mathbb{P}( \max_{ i \in [n] } \{ V_i \} > \tau_1 ) } \\
			&= \frac{ g_{ Nm, \alpha_1 }( k ) }{Nm} \mathbb{P}( \max_{ i \in [n] } \{ V_i \} > t ) \,,
\end{align*}
which is an improvement over using the union bound (as in the previous section). All in all, we have proven the following bound.
\begin{align*}
	 \sum\limits_{ i \in [n] } \mathbb{P}( V_i > t )
		\sum\limits_{ k > ( j - 1 )N }^{ Nm }
		\frac{ \mathbb{P}( T \geq k | \sigma_k = i ) }{ Nm } 
			\geq \left[ \sum\limits_{ k = ( j - 1 )N + 1 }^{ Nm }
				\mathbb{P}( T > k ) \frac{ g_{ Nm, \alpha_1 }( k ) }{Nm} \right]
				\mathbb{P}( \max\limits_{ i \in [n]} \{ V_i \} > t ) \,.			
\end{align*}
Moreover,
\begin{align*}
	\sum\limits_{ k = ( j - 1 )N + 1 }^{ Nm } \mathbb{P}( T > k ) \frac{g_{ Nm, \alpha_1 }( k )}{Nm}
		&= \sum_{ l = j }^{ m } \sum_{ k = 1 }^{ N }
			\mathbb{P}( T > (l-1)N + k ) \frac{g_{ Nm, \alpha_1 }( (l-1)N + k )}{Nm}  \\
		&\geq \sum_{ l = j }^{ m } 
			\left( \prod_{ l' = 1 }^{ l-1 } \alpha_{l'} \right)^{ \frac{1}{m} }
			\sum\limits_{ k = 1 }^{ N } \left( \alpha_l^{ \frac{1}{Nm} } \right)^{ k }
			\frac{g_{ Nm, \alpha_1 }( (l-1)N )}{Nm} \\
		&= \sum\limits_{ l = j }^{ m } 
			\left( \prod\limits_{ l' = 1 }^{ l-1 } \alpha_{l'} \right)^{ \frac{1}{m} }
			g_{ m, \alpha_1 }( l-1 )  
			\frac{ \alpha_l^{ \frac{1}{Nm} } }{ Nm } 
				\frac{ 1 - \alpha_l^{ \frac{ 1 }{ m } } }{ 1 - \alpha_l^{ \frac{1}{Nm} } } \\
		&\xrightarrow[ N \to \infty ]{}
			\sum\limits_{ k = j }^{ m } 
				\left( \prod\limits_{ l \in [k-1] } \alpha_l \right)^{ \frac{1}{m} } 
				g_{ m, \alpha_1 }( k-1 ) 
				\left( \frac{ 1 - \alpha_k^{ \frac{1}{m} } }{ -\ln \alpha_k } \right) \,.
\end{align*}
Putting these two inequalities together, we can conclude the case $j \in \{2, \ldots, m\}$.
	
Lastly, for $j = 1$ (i.e. : $t \in [ \tau_1 , \infty )$) we do as before, in Section \ref{beat}, and use Lemma \ref{Inequality: Conditional stopping time distribution}, Lemma \ref{Inequality: Stopping time distribution} and the union bound to derive the following
\begin{align*}
	\mathbb{P}( V_{\sigma_T} > t ) 
		&= \frac{1}{n} \sum\limits_{ i \in [n] } \mathbb{P}( V_i > t )
				\sum\limits_{ k = 1 }^{ Nm }
				\mathbb{P}( T \geq k | \sigma_k = i ) \\
		&\geq \left[ \frac{1}{n} \sum\limits_{ k = 1 }^{ Nm }
				\mathbb{P}( T > k )
			\right] \mathbb{P}( \max\limits_{ i \in [n]} \{ V_i \} > t ) \\
		&\geq \left[ \sum\limits_{ k = 1 }^{ m } 
				\left( \prod\limits_{ l \in [k-1] } \alpha_l \right)^{ \frac{1}{m} } 
				\left( \frac{ 1 - \alpha_k^{ \frac{1}{m} } }{ -\ln \alpha_k } \right)
			\right] \mathbb{P}( \max\limits_{ i \in [n]} \{ V_i \} > t ) \,,
\end{align*}
where the last inequality is only valid in the limit as $N \to \infty$.
\hfill \qed
\end{proof}

With the previous lemma we can easily establish the improved guarantee. Indeed, take the right-hand-side of the expression in Lemma \ref{Performance guarantee for piece-wise constant blind strategies} and optimize over the choice of $\alpha_1, \ldots, \alpha_m$. We do this optimization numerically and find a particular collection $\alpha_1, \ldots, \alpha_m$ such that the guarantee evaluates to $0.669$, as stated in the following corollary. We must note however that there might be other choices leading to slightly improved guarantees.
\begin{Corollary}
\label{Performance guarantee for piece-wise constant blind strategies, numeric}
	There exists $1 \geq \alpha_1 \geq \ldots \geq \alpha_m \geq 0$ such that
	\[
	\frac{ \mathbb{E}( V_{\sigma_T} ) }{ \mathbb{E}( \max\limits_{ i \in [n] } \{ V_i \} ) }
		\geq 0.66975 \,,
	\]
where $T$ is the stopping time corresponding to the blind strategy $\alpha = \alpha_{\alpha_1, \ldots, \alpha_m}$.
\end{Corollary}

In particular, taking $m = 30$ was enough to derive this result. 

section{A $0.675$ upper bound for blind strategies: proof of Theorem \ref{blindupper}}

In order to prove Theorem \ref{blindupper}, we consider two instances and show that no blind strategy can guarantee better than $0.675$ for both instances.

The first instance consists simply in a single random variable which is nearly deterministic, given by $V_1\sim U(1-\varepsilon,1+\varepsilon)$. 
The second instance has $n$ i.i.d. random variables defined by (and we take $n\to \infty$):
	\[
	V_i \sim \left\{ \begin{array}{l l}
		1/\varepsilon 
			&\text{w.p.} \quad \varepsilon \\
		U( 0, \varepsilon ) 
			&\text{w.p.} \quad 1 - \varepsilon \,.
		\end{array}
		\right.
	\]
Combining these two instances one can show the following result.
\begin{Lem}Let $T$ be the stopping time corresponding to the blind strategy given by $\alpha$. Then
\label{Upper bound for blind strategies}
\begin{align*}
	\sup_{ \alpha } \inf_{ n; F_1, ..., F_n } 
	\frac{ \mathbb{E}( V_{\sigma_T} ) }{ \mathbb{E}( \max\limits_{ i \in [n] } \{ V_i \} ) } 
		\leq \sup_{ \alpha } \min \left\{
			1 - \int\limits_{ 0 }^{ 1 } \alpha( s ) ds \quad , \quad
			\int\limits_{ 0 }^{ 1 } e^{ \int\limits_{ 0 }^{ s } \ln \alpha( w ) dw } ds
			\right\} \,.
\end{align*}
\end{Lem}
With this result we need to compute the quantity on the right-hand-side of the previous lemma to obtain an upper bound on the performance guarantee of any blind strategy. This is done using optimal control theory. The basic procedure consists first in proving that the supremum right-hand-side in Lemma \ref{Upper bound for blind strategies} is attained, then we have Mayer's optimal control problem for which the necessary optimality conditions can be expressed as an integro-differential equation. We conclude by solving this equation numerically, and thus have the following result.	
\begin{Corollary} Let $T$ be the stopping time corresponding to the blind strategy given by $\alpha$. Then
\label{numeric}
	\[
	\sup\limits_{ \alpha } \inf\limits_{ n; F_1, ..., F_n } 
	\frac{ \mathbb{E}( V_{\sigma_T} ) }{ \mathbb{E}( \max\limits_{ i \in [n] } \{ V_i \} ) }
		\leq 0.675 \,.
	\]
\end{Corollary}

\begin{proof}[Proof of Lemma \ref{Upper bound for blind strategies}]
	The first instance is $V_1 = U(1-\varepsilon, 1+\varepsilon)$, with $\varepsilon > 0$. Notice that, $\tau = 1 - \varepsilon + \alpha( u ) 2 \varepsilon$,  where $u \sim U(0, 1)$ and, by direct computation, we have that
\begin{align*}
	\mathbb{E}( V_{\sigma_T} ) 
		&= \int\limits_{ 0 }^{ 1 } \mathbb{E}( V_{\sigma_T} | u = s ) ds \\
		&= \int\limits_{ 0 }^{ 1 } \left( 1 + \varepsilon  \left( \alpha( s ) - \frac{1}{2} \right)  \right) 
			( 1 - \alpha(s) ) ds \\
		&\xrightarrow[ \varepsilon \to 0]{} 
			\int\limits_{ 0 }^{ 1 } ( 1 - \alpha(s) ) ds \,.
\end{align*}

	The second instance has $n$ i.i.d. random variables defined by:
	\[
	V_i \sim \left\{ \begin{array}{l l}
		\frac{ 1 }{ \varepsilon } 
			&w.p. \quad \varepsilon \\
		U( 0, \varepsilon ) 
			&w.p. \quad 1 - \varepsilon \,.
		\end{array}
		\right.
	\]
Moreover, for $\varepsilon$ small enough,
	\[
	\mathbb{P}( \max\limits_{ i \in [n] } \{ V_i \} \leq t ) = \left\{ \begin{array}{l l}
		0
			&; t < 0 \\
		\left( \frac{ 1 - \varepsilon }{ \varepsilon }\right)^{ n } t^{ n }
			&; 0 \leq t < \varepsilon \\
		( 1 - \varepsilon )^{ n }
			&; \varepsilon \leq t < \frac{ 1 }{ \varepsilon } \\
		1 
			&; \frac{ 1 }{ \varepsilon } \leq t 
		\end{array} 
		\right.
	\] 
Notice that we can assume $\alpha( x ) < 1$, for $x > 0$, since there is no gain in always rejecting. Then, $\tau_i = \sqrt[ n ]{ \alpha( u_{ [i] } ) } \frac{ \varepsilon }{ 1 - \varepsilon }$ and we have that $\mathbb{P}( V_i \leq \tau_i ) = \sqrt[ n ]{ \alpha( u_{ [i] }) }$. By direct computation,
\begin{align*}
	\lim\limits_{ \varepsilon \to 0 } \mathbb{E}( V_{\sigma_T} | u )
		&= 1 + \sqrt[ n ]{ \alpha( u_{ [1] }) }
			+ \sqrt[ n ]{ \alpha( u_{ [1] }) \alpha( u_{ [2] }) } 
		+ \ldots + \sqrt[ n ]{ \alpha( u_{ [1] }) ... \alpha( u_{ [n-1] }) } \\
		&= \sum\limits_{ i = 0 }^{ n - 1 } \prod\limits_{ j = 1 }^{ i } 
			\sqrt[ n ]{ \alpha( u_{ [j] }) } \,.
\end{align*}
In addition, we have $\mathbb{E}( \max\limits_{ i \in [n] } \{ V_i \} ) = \frac{ 1 }{ \varepsilon } ( 1 - \left(1 - \varepsilon \right)^{ n } ) \xrightarrow[ \varepsilon \to 0 ]{} n$. Then, 
\begin{align*}
	\lim\limits_{ n \to \infty } \lim\limits_{ \varepsilon \to 0 }
	\frac{ \mathbb{E}( V_{\sigma_T} ) }{ \mathbb{E}( \max\limits_{ i \in [n] } \{ V_i \} ) }
		 &= \lim\limits_{ n \to \infty } \mathbb{E}_{u} \left[
		 	\frac{1}{n} \sum\limits_{ i = 0 }^{ n - 1 } 
		 		\prod\limits_{ j = 1 }^{ i } \sqrt[ n ]{ \alpha( u_{ [j] }) } \right] 
		 = \int\limits_{ 0 }^{ 1 } e^{ \int\limits_{ 0 }^{ s } \ln \alpha( w ) dw } ds \,.
\end{align*}
\hfill \qed
\end{proof}

\begin{proof}[Sketch of Proof of Corollary \ref{numeric}]
We first show that the supremum given by Lemma \ref{Upper bound for blind strategies} is attained at a certain $\alpha^*$. To this end we note that, without loss of generality, we can consider the supremum over nonincreasing functions $\alpha$, by a simple exchange of mass argument. Then, we note that the set of nonincreasing functions from $[0,1]$ to itself is compact for the $||\cdot||_\infty$ and the functional being optimized is continuous for that metric. Then we deduce the existence of $\alpha^*$, and furthermore it satisfies
 $$1 - \int_{ 0 }^{ 1 } \alpha^*( s ) ds = \int_{ 0 }^{ 1 } \exp\left[ \int_{ 0 }^{ s } \ln \alpha^*( w ) dw \right] ds.
 $$
 Therefore, $\alpha^*$ is the solution of the following optimal control problem:
	\[
	(P) \left\{ \begin{array}{l l}
		\min\limits_{ \alpha } 
			& - x_1( 1 )	= - \int\limits_{ 0 }^{ 1 } 1 - \alpha( t ) dt \\
		s.t.:
			& \dot{ x }( t ) 
				= f( x(t), \alpha(t) ) 
				= \left( \begin{array}{c}
					1 - \alpha( t ) \\
					\exp[ x_3( t ) ] \\
					\ln \alpha( t )
					\end{array}
					\right) \\
			& x(0) = ( 0, 0, 0 )' \\
			& ( t, x(t) ) \in [0, 1] \times \mathbb{R}^3 \\
			& \alpha( t ) \in [0, 1] \\			
			& x_1( 1 ) = x_2( 1 ) \,.
		\end{array}
		\right.
	\]
This is generally called a Mayer problem and the necessary optimality conditions (Pontryagin maximum principle) leads to identify $\alpha^*$ with $\alpha_{ K, \overline{t} }$ defined by
	\[
		\alpha_{ K, \overline{t} }( t ) = \left\{ \begin{array}{l l}
		1
			&; 0 \leq t < \overline{t} \\
		\beta_{ K, \overline{t} }( t )
			&; \overline{t} \leq t \leq 1,
		\end{array}
		\right.	
	\]
where $K \in [0, 3]$ and $\overline{t} \in [0, 1/3]$ and $\beta_{ K, \overline{t} }$ is the solution of the following integro-differential equation
\begin{align}
\label{ODE}
	\left\{ \begin{array}{l l}
		\dot{\beta}( t ) = 
		- K \exp\left[ \int\limits_{ \overline{t} }^{ t } \ln \beta( s ) ds \right]
			& t \in ( \overline{t}, 1) \\
		\beta( 1 ) = 0 \,.
		\end{array}
		\right.
\end{align}

To solve numerically this equation, consider the change of variables $g(t) = \int_{ \overline{t} }^{ t } \ln \beta( s ) ds$, so that the equation becomes the second order ODE 
	\[
	\left\{ \begin{array}{l l}
		e^{ \dot{ g }( t ) } \ddot{ g }( t ) 
			= - K e^{ g( t ) } 
			&; \quad t \in ( \overline{t}, 1 ) \\
		g( \overline{t} ) 
			= 0 \\
		\dot{g}( \overline{t} ) 
			= \ln \beta( \overline{t} ) \,.
		\end{array}
		\right.
	\]
Because $\exp( \cdot )$ is continuous and locally Lipschitz, this is a well-posed Cauchy problem with a unique local solution. The initial condition $\dot{g}( \overline{t} ) = \ln \beta( \overline{t} )$ turns out to be simply a replacement for $\dot{g}( 1 ) = - \infty$ in the sense that we search for the solutions $g$ such that $g( \overline{t} ) = 0$ and exploits at time $1$. This seemingly numerical difficulty is well treated using solvers for stiff ODE such as \textit{ode15s} of \textit{Matlab}.

	Then, we numerically compute \eqref{ODE} to determine that
\begin{align*}
	\sup\limits_{ \alpha } \left\{
			1 - \int\limits_{ 0 }^{ 1 } \alpha( s ) ds , 
			\int\limits_{ 0 }^{ 1 } e^{ \int\limits_{ 0 }^{ s } \ln \alpha( w ) dw } ds
			\right\} 
		&= \sup\limits_{ \substack{ K \in [0, 3] \\ \overline{t} \in [0, 1/3] } } \left\{
			1 - \int\limits_{ 0 }^{ 1 } \alpha_{ K, \overline{t} }( s ) ds , 
			\int\limits_{ 0 }^{ 1 } e^{ \int\limits_{ 0 }^{ s } \ln \alpha_{ K, \overline{t} }( w ) dw } ds
			\right\} \\
		&\leq 0.675 \,.
\end{align*}
Finally we note that if \eqref{ODE} has no solution, this simply means that $\alpha^*$ does not corresponds to $\alpha_{ K, \overline{t} }$ and thus it is not taken into account in the previous supremum. 
\hfill \qed
\end{proof}

\section{An upper bound on the performance of general strategies}
\label{sec:naupper}
In this section, we provide a $\sqrt{3}-1$ upper bound on the performance of any strategy, thus proving Theorem \ref{naupper}. This implies a separation between the i.i.d. prophet inequalities and 
the prophet secretary problem. 
Surprisingly, our bound comes from analyzing the following simple instance. 

Take $a \in [0, 1]$ and consider $n+1$ random variables whose values are distributed as:
\begin{align*}
	V_1, \ldots, V_n &\sim \left\{ \begin{array}{l l}
		n 	&w.p. \quad \frac{1}{n^2}\\
		0 	&w.p. \quad 1 - \frac{1}{n^2}
	\end{array}
	\right. \\
	V_{n+1} &\equiv a \,.
\end{align*}

Clearly any reasonable algorithm would always accept a value of $n$ and never accept a value of $0$. Therefore the only decision an algorithm has to make is that of whether accepting a value of $a$ or not. Solving the dynamic programming, the algorithms picks the value $a$ if and only if it appears at time $j \in [n+1]$ or later. From $j$ onwards, the expectation of the future is less than the value $a$. Let $\sigma$ be the random permutation and $T$ be the implied stopping time, then we have that for $i = 1, \ldots, j-1$
	\[
	\mathbb{E}( V_{\sigma_T} | \sigma_{i} = n+1 ) 
		= n \left[ 1 - \left(1 - \frac{1}{n^2} \right)^n \right] \,. 
	\]
On the other hand, for $i = j, \ldots, n+1$, the following holds
	\[
	\mathbb{E}( V_{\sigma_T} | \sigma_{i} = n+1 ) 
		= n \left[ 1 - \left(1 - \frac{1}{n^2} \right)^{i-1} \right] + \left(1 - \frac{1}{n^2} \right)^{i-1} a \,.
	\]
Therefore, 
	\begin{eqnarray*}
	\mathbb{E}( V_{\sigma_T} )
		&=& \frac{j - 1}{n+1} n \left[ 1 - \left(1 - \frac{1}{n^2} \right)^n \right]
			+ \frac{1}{n+1} \sum_{i = j}^{n+1}  n \left[ 1 - \left(1 - \frac{1}{n^2} \right)^{i-1} \right] + \left(1 - \frac{1}{n^2} \right)^{i-1} a \,
			\\
			&\leq&
			\frac{j - 1}{n+1} +\frac{1}{n+1} \sum_{i = j}^{n+1}  \left[\frac{i-1}{n} +a \right]
			\\
			& = & \frac{j}{n}(1-a)+a+\frac{1}{2}-\frac{1}{2}\left(\frac{j}{n}\right)^2+O\left(\frac{1}{n}\right)
			\\
			&\leq&
			1+\frac{a^2}{2}+O\left(\frac{1}{n}\right)
	\end{eqnarray*}
	where the first inequality stems from the fact that $n \left[1 - \left(1 - \frac{1}{n^2} \right)^{i-1} \right] \leq \frac{i-1}{n} $, and the second follows by optimizing over $j$. 

%
On the other hand
	\[
	\mathbb{E}( \max_{i\in [n]} \{ V_i \} ) 
		= n \left[ 1 - \left(1 - \frac{1}{n^2} \right)^n \right] + \left(1 - \frac{1}{n^2} \right)^n a
		\xrightarrow[n \to \infty]{} 1 + a \,.
	\]
We notice that choosing $a = \sqrt{3} - 1$ we get that
	\[
	\limsup_{n \to \infty} \frac{ \mathbb{E}( V_{ \sigma_T } ) }{ \mathbb{E}( \max\limits_{ i \in [n]} \left\{ V_i \right\} ) }
		\leq \frac{ 1 + \frac{a^2}{2} }{ 1 + a }
		= \sqrt{3} - 1 \approx 0.732 \,.
	\]

\section{Dealing with discontinuous distributions}
\label{discontinuous}
	
	In this section we explain how to use a \textit{blind strategy} in instances where the distributions $F_1, \ldots, F_n$ are not necessarily continuous. Recall that in the definition of blind strategies in Section \ref{prelim}, we need the existence of $\tau_i$ such that $\mathbb{P}( \max_{ i \in [n] } \{ V_i \} \leq \tau_i ) = \alpha( u_{[i]} )$. So, what happens if such thresholds $\tau_1, ..., \tau_n$ do not exist? For the purpose of studying the prophet inequality, the performance of a strategy defined over instances with continuous distributions is always extendable to discontinuous ones allowing stochastic tie breaking. In this case, we can explicitly define the strategy that $\alpha$ induces over discontinuous instances. The resulting strategy no longer depends on the distribution of the maximum only.
	
	The procedure to compute the tie breaking is quite natural:
\begin{enumerate}
	\item Approximate the instance.
	\item Study the strategy induced by $\alpha$ in the approximated instance.
	\item Replicate what would happen in the original instance, allowing tie breaking.
\end{enumerate}
	Given a realization of uniform random variables $u_1, \ldots, u_n$, assume that $\tau_i$ does not exist, in other words, for some $i \in [n]$, there is a $\tau \in \mathbb{R}$ such that
	\[
	\lim\limits_{ \varepsilon \to 0} \mathbb{P}( \max\limits_{ i \in [n] } \{ V_i \} \leq \tau - \varepsilon )
		< \alpha( u_{[i]} ) 
		< \mathbb{P}( \max\limits_{ i \in [n] } \{ V_i \} \leq \tau ) \,.
	\]
The stochastic tie breaking consists in accepting the value $\tau$ with some probability, say $p_i$. This acceptance rate depends on the whole instance, not only on the distribution of the maximum, and on the identity of the revealed variable. To compute these acceptance rates we use the following procedure. For $\varepsilon > 0$, consider the following approximated instance
	\[
	F_i^{\varepsilon} ( t ) = \left\{ \begin{array}{l l}
		F_i( t ) \\
		F_i( \tau - \varepsilon ) 
			+ \frac{ t - \tau + \varepsilon }{ \varepsilon } \left( F( \tau ) - F_i( \tau - \varepsilon \right))
		\end{array}
		\right. \,,
	\]
for $t \not \in [ \tau - \varepsilon, \tau]$ in the first case and $t \in [ \tau - \varepsilon, \tau] $ in the second case.
This instance has a continuous distribution of the maximum in $[ \tau - \varepsilon, \tau]$ and we are able to find $\tau^{\varepsilon}$, the corresponding threshold for the approximated instance, such that
	\[
	\mathbb{P}( \max\limits_{ i \in [n] } \{ V_i \} \leq \tau^{\varepsilon} ) = \alpha( u_{[i]} ) \,.
	\]
Then, we compute, for $j \in [n]$, $\beta_j := \lim_{ \varepsilon \to 0 } F_j^{\varepsilon}( \tau^{\varepsilon} )$. To finish, we define, for $j$ such that $\mathbb{P}( V_j = \tau ) > 0$,
	\[
	p_i( j ; F_1, ..., F_n ) := \frac{ F_j( \tau ) - \beta_j }{ \mathbb{P}( V_j = \tau ) }
	\]
and $p_i = 0$ otherwise. This will induce that, faced with $V_j$ at time $i$, the gambler accepts its realization with probability $1 - \beta_j$. To be more precise, we use the following procedure.
\begin{algorithm}[H]
\caption{Stochastic TTA}
\label{Alg: TFTA}
\begin{algorithmic}[1]
	\FOR{ i = 1, ..., $n$ }
		\IF{ $V_{\sigma_i} > \tau_i$ }
			\STATE Take $V_{\sigma_i}$
		\ELSIF{ $V_{\sigma_i} = \tau_i$ }
			\STATE Take $V_{\sigma_i}$ with probability $p_i( \sigma_i ; F_1, ..., F_n )$
		\ENDIF
	\ENDFOR
\end{algorithmic}
\end{algorithm}
With this procedure, it is easy to see that all results extend to general instances.


\begin{thebibliography}{99}






\bibitem{AE17} 
M. Abolhassani, S. Ehsani, H. Esfandiari, M.T. Hajiaghayi, R. Kleinberg, B. Lucier,
 {\em 
 }, 
 in Proceedings of the 49th Annual ACM SIGACT Symposium on Theory of Computing, 
 STOC, 2017, pp.~61-71.

\bibitem{Alaei2015}
S. Alaei, J. Hartline, R. Niazadeh, E. Pountourakis, Y. Yuan, 
 {\em Optimal Auctions vs. Anonymous Pricing}, 
 in Proceedings of the 56th Annual Symposium on Foundations of Computer Science,
 FOCS, 2015, pp.~1446-1463 .

\bibitem{ACK18} 
Y. Azar, A. Chiplunkar, H. Kaplan, 
 {\em Prophet Secretary: Surpassing the $1-1/e$ Barrier}, 
 in Proceedings of the ACM Conference on Economics and Computation, 
 EC, 2018, pp.~303-318.

\bibitem{BGP18} 
H. Beyhaghi, N. Golrezaei, R. Paes Leme, M. Pal, B. Sivan. 
 {\em Improved Approximations for Posted Price and Second Price Mechanisms}. 
 Unpublished manuscript 2018. \url{http://www.cs.cornell.edu/~hedyeh/papers/posted_prices.pdf}

\bibitem{CHM10}
S. Chawla, J. Hartline, D.L. Malec, B. Sivan, 
 {\em Multi-parameter Mechanism Design and Sequential Posted Pricing}, 
 in Proceedings of the 42nd Annual ACM SIGACT Symposium on Theory of Computing, 
 STOC, 2010, pp.~311-320.

\bibitem{CFH17}
J. Correa, P. Foncea, R. Hoeksma, T. Oosterwijk, T. Vredeveld, 
 {\em Posted price mechanisms for a random stream of customers}, 
 in Proceedings of the ACM Conference on Economics and Computation, 
 EC, 2017, pp.~169-186.

\bibitem{CFPV18} 
J. Correa, P. Foncea, D. Pizarro, V. Verdugo, 
 {\em From pricing to prophets and back!}, 
 Unpublished manuscript 2018.

\bibitem{EINAV}
L. Einav, C. Farronato, J. Levin, N. Sundaresan, 
 {\em Auctions versus posted prices in online markets}, 
 Journal of Political Economy, 
 126(1) (2018), pp.~178-215.

\bibitem{EHL15}
H. Esfandiari, M. Hajiaghayi, V. Liaghat, M. Monemizadeh, 
 {\em Prophet Secretary}, 
 in Proceedings of the 23rd Annual European Symposium,
 ESA, 2015, pp.~496-508.

\bibitem{EHK18} 
S. Ehsani, M. Hajiaghayi, T. Kesselheim, S. Singla, 
 {\em Prophet Secretary for Combinatorial Auctions and Matroids},
 in Proceedings of the 29th Annual ACM-SIAM Symposium on Discrete Algorithms,
 SODA, 2018, pp.~700-714.

\bibitem{GM66}
J.P. Gilbert, F. Mosteller, 
 {\em Recognizing the maximum of a sequence},
 Journal of the American Statistical Association, 
 61(313) (1966), pp.~35-76.

\bibitem{HKS2007}
M. Hajiaghayi, R. Kleinberg, T. Sandholm, 
 {\em Automated online mechanism design and prophet inequalities}, 
 in Proceedings of the 22nd AAAI Conference on Artificial Intelligence,
 AAAI, 2007, pp.~58-65.

\bibitem{HK82}
T.P. Hill, R.P. Kertz, 
 {\em Comparisons of stop rule and supremum expectations of i.i.d. random variables}, 
 The Annals of Probability, 
 10(2) (1982), pp.~336-345.

\bibitem{K86}
R.P. Kertz, 
 {\em Stop rule and supremum expectations of i.i.d. random variables: A complete comparison by conjugate duality},
 Journal of Multivariate Analysis, 
 19(1) (1986), pp.~88-112.

\bibitem{KW12}
R. Kleinberg, S.M. Weinberg, 
 {\em Matroid prophet inequalities}, 
 in Proceedings of the 44th Annual ACM SIGACT Symposium on Theory of Computing, 
 STOC, 2012, pp.~123-136.

\bibitem{KS77}
U. Krengel, L. Sucheston, 
 {\em Semiamarts and finite values}, 
 Bull. Amer. Math. Soc, 
 83(4) (1977), pp.~745-747.

\bibitem{KS78}
U. Krengel, L. Sucheston, 
 {\em On semiamarts, amarts, and processes with finite value}, 
 Adv. in Probability, 
 4 (1978), pp.~197-266.

\bibitem{MOA11} 
A.W. Marshall, I. Olkin, B.C. Arnold, 
 {\em Inequalities: Theory of Majorization and Its Applications}, 
 Springer, New York, 2011.

\bibitem{M81}
R.B. Myerson, 
 {\em Optimal Auction Design},
 Mathematics of Operations Research, 
 6(1) (1981), p.~58-73.

\bibitem{PPT92} 
J. Pecaric, F. Proshman, Y. Tong, 
 {\em Convex Functions, Partial Orderings, and Statistical Applications},
 Academic Press, 1992.

\bibitem{SC83}
E. Samuel-Cahn, 
 {\em Comparisons of threshold stop rule and maximum for independent nonnegative random variables},
 The Annals of Probability, 
 12(4) (1983), pp.~1213-1216.

\bibitem{Y11}
Q. Yan, 
 {\em Mechanism design via correlation gap},
 in Proceedings of the 22nd Annual ACM-SIAM Symposium on Discrete Algorithms,
 SODA, 2011, pp.~710-719.

\end{thebibliography}
\end{document}